\documentclass{article}
\usepackage{amsmath}
\usepackage{amsmath,amssymb}
\usepackage{graphicx}
\usepackage{epstopdf}
\usepackage{epsfig}

\newtheorem{theorem}{Theorem}

\newtheorem{definition}[theorem]{Definition}

\newtheorem{lemma}[theorem]{Lemma}

\newtheorem{proposition}[theorem]{Proposition}
\newtheorem{remark}[theorem]{Remark}

\newenvironment{proof}[1][Proof]{\textbf{#1.} }{\ \rule{0.5em}{0.5em}}

\begin{document}

\title{Non-abelian Weyl Commutation Relations and the Series Product of Quantum
Stochastic Evolutions}
\author{D. Gwion Evans\thanks{%
Institute for Mathematics and Physics, Aberystwyth University, SY23 3BZ,
Wales, United Kingdom. Email: dfe@aber.ac.uk.}, John E. Gough\thanks{
Institute for Mathematics and Physics, Aberystwyth University, SY23 3BZ,
Wales, United Kingdom. Email: jug@aber.ac.uk.}, Matthew R. James\thanks{%
ARC Centre for Quantum Computation and Communication Technology, Research
School of Engineering, Australian National University, Canberra, ACT 0200,
Australia. Email: Matthew.James@anu.edu.au.}}
\date{}
\maketitle

\begin{abstract}

We show that the series product, which serves as an algebraic rule for
connecting state-based input/output systems, is intimately related to the
Heisenberg group and the canonical commutation relations. The series product
for quantum stochastic models then corresponds to a non-abelian generalization of
the Weyl commutation relation. We show that the series product gives  the
general rule for combining the generators of quantum stochastic evolutions using
a Lie-Trotter product formula.
\end{abstract}

\section{Introduction}
The aim of this paper is to make some striking connections between the rules
for combining models in series in control system theory and the Weyl
commutation relations. In the process, we develop a more intrinsic view of the unitary adapted
processes of Hudson and Parthasarathy \cite{HP} as non-abelian versions of the
Weyl unitaries - where the non-abelian nature arises from the presence of the initial space.
Our starting point is a surprising connection between
the theory of classical linear state space models and the canonical
commutation relations.

\subsection{State-Based Input/Output Systems}
Let $\mathcal{X}, \mathcal{U} $ and $\mathcal{Y}$ be finite dimensional vector spaces over the reals.
A controlled flow on the state space $\mathcal{X}$ is given by the
dynamical equations
\begin{equation*}
\dot{x}=v\left( x,u\right)
\end{equation*}
where $u$ is a $\mathcal{U}$-valued function of time called the input process. An output process $y$ taking values in $\mathcal{Y}$ is given by some relation of the general form

\begin{equation*}
y= h\left( x,u\right).
\end{equation*}

The situation is sketched in figure \ref{fig:1}, along with the case where we further decompose the value spaces into subspaces.

\begin{figure}[tbph]
\centering
\includegraphics[width=1.00\textwidth]{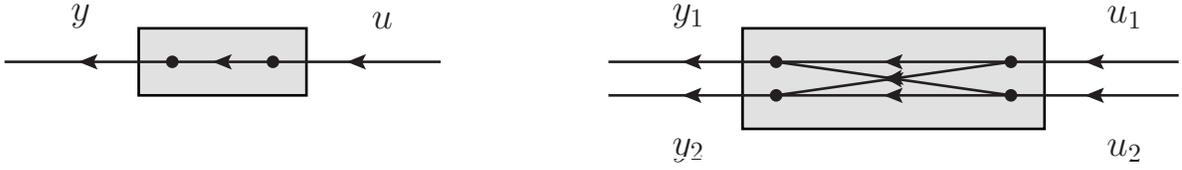}
\caption{The left-hand picture sketches an input-state-output model $\left(
\mathcal{U},\mathcal{X},\mathcal{Y} \right) $ corresponding to the system of equations (\ref{eq:simple}).
On the right we consider decompositions of the input and output value spaces,
 $\mathcal{U}=\mathcal{U}_1 \oplus \mathcal{U}_2$ and $\mathcal{Y} =\mathcal{Y}_1 \oplus \mathcal{Y}_2$ respectively.}
\label{fig:1}
\end{figure}

\subsection{Linear Systems}

We consider a vector input $u\left( \cdot \right) $ leading to a vector
output $y\left( \cdot \right) $ according to the model
\begin{equation}
\left\{
\begin{array}{c}
\dot{x}=Kx+Lu; \\
y=Mx+Nu;
\end{array}
\right.  \label{eq:simple}
\end{equation}
or
\begin{equation*}
\left[
\begin{array}{c}
\dot{x} \\
y
\end{array}
\right] =\mathbf{V}\left[
\begin{array}{c}
x \\
u
\end{array}
\right] \text{, where }\mathbf{V}=\left[
\begin{array}{cc}
K & M \\
L & N
\end{array}
\right] .
\end{equation*}
Here $x\left( \cdot \right) $ is the state vector state, initialized at some
value $x_{0}$, and $\mathbf{V}$ is referred to as the \textit{model matrix}
for the model. For $u\left( \cdot \right) $ integrable, the solution can be
written immediately as $y\left( t\right) =Nu\left( t\right)
+\int_{0}^{t}Me^{K\left( t-s\right) }Lu(s)ds+Me^{Kt}x_{0}$: we also note
that the input-output relation is described by the transfer function $%
T\left( s\right) =N+M\left( sI-K\right) ^{-1}L$ which is determined from the
model matrix. The situation is sketched in the top left picture in figure
\ref{fig:lie_trotter0b}.

As the inputs and outputs are vector-valued they may be further decomposed
as say $u=\left[
\begin{array}{c}
u_{1} \\
u_{2}
\end{array}
\right] $ and $y=\left[
\begin{array}{c}
y_{1} \\
y_{2}
\end{array}
\right] $. This is sketched on the right in figure \ref
{fig:lie_trotter0b}. The model matrix is then
\begin{equation}
\mathbf{V}=\left[
\begin{array}{cc}
K & [M_{1},M_{2}] \\
\left[
\begin{array}{c}
L_{1} \\
L_{2}
\end{array}
\right] & \left[
\begin{array}{cc}
N_{11} & N_{12} \\
N_{21} & N_{22}
\end{array}
\right]
\end{array}
\right] .  \label{eq:model_matrix_decomp}
\end{equation}
In each case we have a port for each input/output. The lines external to the
block represent an input or output, while the lines internal to the block
correspond to a non-zero entry $N_{ij}$ connect input port $j$ to output
port $i$. The picture on the bottom of figure \ref{fig:lie_trotter0b}
sketches the situation where $N_{12}=N_{21}=0$.

\begin{figure}[tbph]
\centering
\includegraphics[width=0.50\textwidth]{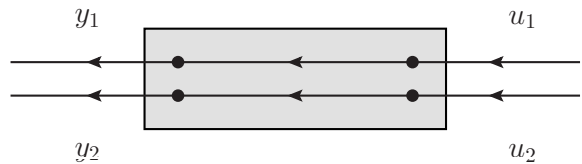}
\caption{The picture sketches the situation in (\ref{eq:model_matrix_decomp}) where $N_{12}=N_{21}=0$.}
\label{fig:lie_trotter0b}
\end{figure}

\subsection{Concatenation}

Suppose that we have a pair of such models with the same state space (with
variable $x$) and model matrices $\mathbf{V}_{i}=\left[
\begin{array}{cc}
K_{i} & M_{i} \\
L_{i} & N_{i}
\end{array}
\right] $, that is,
\begin{equation*}
\left[
\begin{array}{c}
\dot{x} \\
y_{2}
\end{array}
\right] =\mathbf{V}_{2}\left[
\begin{array}{c}
x \\
u_{2}
\end{array}
\right] ,\quad \left[
\begin{array}{c}
\dot{x} \\
y_{1}
\end{array}
\right] =\mathbf{V}_{1}\left[
\begin{array}{c}
x \\
u_{1}
\end{array}
\right] .
\end{equation*}

We may superimpose the two models to get the \textit{concatenated} model
\begin{equation*}
\left\{
\begin{array}{l}
\dot{x}=\left( K_{1}+K_{2}\right) x+M_{1}u_{1}+M_{2}u_{2}, \\
y_{1}=L_{1}x+N_{1}u_{1}, \\
y_{2}=L_{2}x+N_{2}u_{2},
\end{array}
\right.
\end{equation*}
- writing $v_i (x) = K_i x +M_i u_i $ for the separate state velocity fields
($i=1,2$), the concatenation rule effectively takes the combined velocity
field
\begin{equation}
v(x) =v_1 (x) +v_2 (x).  \label{eq: sum}
\end{equation}
At the level of model matrices, this corresponds to the rule (see figure \ref{fig:lie_trotter0a})
\begin{equation}
\mathbf{V}_{1}\boxplus \mathbf{V}_{2}\triangleq \left[
\begin{array}{ccc}
K_{1}+K_{2} & M_{1} & M_{2} \\
L_{1} & N_{1} & 0 \\
L_{2} & 0 & N_{2}
\end{array}
\right] ,\quad \left[
\begin{array}{c}
\dot{x} \\
y_{1} \\
y_{2}
\end{array}
\right] =\mathbf{V}_{1}\boxplus \mathbf{V}_{2}\left[
\begin{array}{c}
x \\
u_{1} \\
u_{2}
\end{array}
\right] .  \label{eq:concat_sum}
\end{equation}

\begin{figure}[tbph]
\centering
\includegraphics[width=1.00\textwidth]{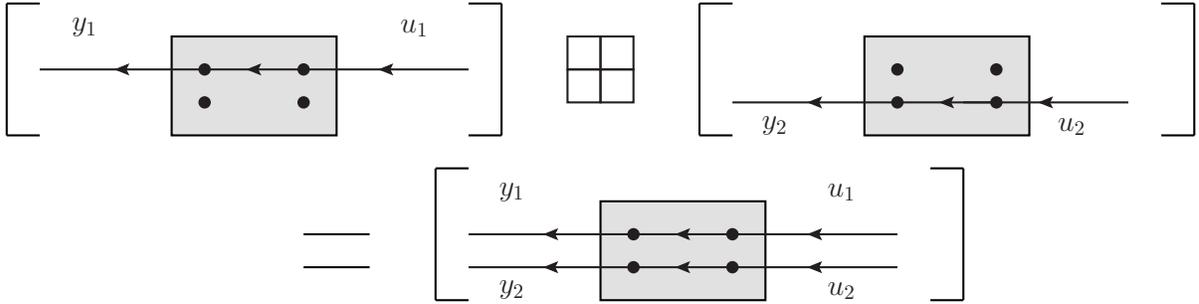}
\caption{The concatenation of two models $\mathbf{V}_{1}\boxplus \mathbf{V}_{2}$ with the same state space $\mathcal{X}$.}
\label{fig:lie_trotter0a}
\end{figure}

The concatenation sum of two model matrices will result in the type of
situation depicted in the picture in figure \ref{fig:lie_trotter0b},
that is, model (\ref{eq:model_matrix_decomp}) with $N_{11}=N_1, \,
N_{22}=N_2 , \, N_{12}=0=N_{21}$.

It is worth remarking that the addition rule (\ref{eq: sum}) makes sense for
stochastic flows, either in the It\={o} or Stratonovich form: here we would have stochastic differential equations
\begin{eqnarray*}
dx &=&v\left( x\right) dt+\sigma \left( x\right) dU \\
dY &=&h\left( x\right) dt+\gamma dU
\end{eqnarray*}
where $U$ is a semi-martingale with $\dot{U}=u$, $\dot{Y}=y$ formally. A
concatenation would then take the form
\begin{eqnarray*}
dx &=&\left[ v_{1}\left( x\right) +v_{2}\left( x\right) \right] dt+\sigma
_{1}\left( x\right) dU_{1}+\sigma _{2}\left( x\right) dU_{2}, \\
dY_{1} &=&h_{1}\left( x\right) dt+\gamma _{1}dU_{1}, \\
dY_{2} &=&h_{2}\left( x\right) dt+\gamma _{2}dU_{2}.
\end{eqnarray*}

\subsection{Series Product}

Following this, (assuming the dimensions match) we may then introduce
feedback into the concatenated model (\ref{eq:concat_sum}) by setting the
output $y_{1}(\cdot )$ of the first system equal to the input $u_{2}(\cdot )$
of the second. Setting $u_{2}=y_{1}(=L_{1}x+N_{1}u_{1})$ and eliminating
these as internal signals in the concatenated system above, we reduce to a
linear system
\begin{equation*}
\left\{
\begin{array}{l}
\dot{x}=\left( K_{1}+K_{2}+M_{2}L_{1}\right) x+\left(
M_{1}+M_{2}N_{1}\right) u_{1}, \\
y_{2}=\left( L_{2}+N_{2}L_{1}\right) x+N_{2}N_{1}u_{1},
\end{array}
\right.
\end{equation*}
with model matrix
\begin{equation}
\mathbf{V}_{2}\ast \mathbf{V}_{1}\triangleq \left[
\begin{array}{cc}
K_{1}+M_{2}L_{1}+K_{2} & M_{1}+M_{2}N_{1} \\
L_{2}+N_{2}L_{1} & N_{2}N_{1}
\end{array}
\right] .  \label{eq:series_model}
\end{equation}

\begin{figure}[tbph]
\centering
\includegraphics[width=0.60\textwidth]{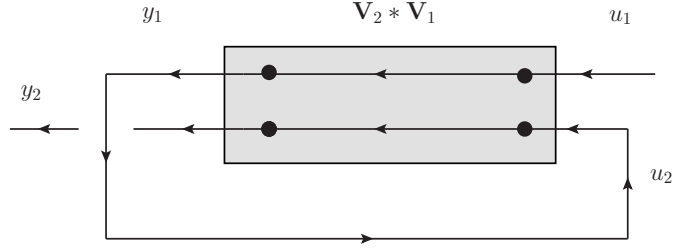}
\caption{We sketch a concatenation of two models where the output $y_{1}$ is
fed back in as input $u_{2}$ to the same system: resulting in a reduced
model $\mathbf{V}_{2}\ast \mathbf{V}_{1}$.}
\label{fig:lie_trotter1}
\end{figure}

We refer to the binary operation $\ast $ as the (general) \textit{series
product}, and this will recur in this paper under various guises.

\subsection{The Heisenberg Group}

The collection of square model matrices of a fixed dimension, and with lower
block $N$ invertible, forms a group with the series product as law. A
straightforward representation $\rho $ of these groups as a subgroup of
higher dimensional upper block-triangular matrices (with the series product
law now replaced by ordinary matrix multiplication) is given by
\begin{equation*}
\rho :\left[
\begin{array}{cc}
K & M \\
L & N
\end{array}
\right] \mapsto \left[
\begin{array}{ccc}
I & M & K \\
0 & N & L \\
0 & 0 & I
\end{array}
\right] .
\end{equation*}

We now make the observation that we have obtained (in the case $N=I$) the
Heisenberg group associated with the canonical commutation relations: we
refer to the situation $N\neq I$ as the extended Heisenberg group. For a
single-input, single-output, single variable system, we see that the Lie
group is generated by
\begin{equation*}
\mathfrak{a}=\left[
\begin{array}{ccc}
0 & 1 & 0 \\
0 & 0 & 0 \\
0 & 0 & 0
\end{array}
\right] ,\mathfrak{a}^{\dag }=\left[
\begin{array}{ccc}
0 & 0 & 0 \\
0 & 0 & 1 \\
0 & 0 & 0
\end{array}
\right] ,\mathfrak{n}=\left[
\begin{array}{ccc}
0 & 0 & 0 \\
0 & 1 & 0 \\
0 & 0 & 0
\end{array}
\right] ,\mathfrak{t}=\left[
\begin{array}{ccc}
0 & 0 & 1 \\
0 & 0 & 0 \\
0 & 0 & 0
\end{array}
\right] ,
\end{equation*}
and we note the product table
\begin{equation*}
\begin{tabular}{l|llll}
$\times $ & $\mathfrak{a}$ & $\mathfrak{n}$ & $\mathfrak{a}^{\dag }$ & $%
\mathfrak{t}$ \\ \hline
$\mathfrak{a}$ & 0 & $\mathfrak{a}$ & $\mathfrak{t}$ & 0 \\
$\mathfrak{n}$ & 0 & $\mathfrak{n}$ & $\mathfrak{a}^{\dag }$ & 0 \\
$\mathfrak{a}^{\dag }$ & 0 & 0 & 0 & 0 \\
$\mathfrak{t}$ & 0 & 0 & 0 & 0
\end{tabular}
\end{equation*}
so that the non-zero Lie brackets are $\left[ \mathfrak{a},\mathfrak{a}%
^{\dag }\right] =\mathfrak{t}$, $\left[ \mathfrak{a},\mathfrak{n}\right] =%
\mathfrak{a}$ and $\left[ \mathfrak{n},\mathfrak{a}^{\dag }\right] =%
\mathfrak{a}^{\dag }$.

\subsection{Cascading}

We should explain that the terminology of ``series'' is meant to driving
fields acting on a given system in series and the use of the single state
variable $x$ allows for the possibility of variable sharing. The situation
where two separate systems connected in series will be termed ``cascading''
and we should emphasize that this is indeed as a special case. Here the
joint state $x=\left[
\begin{array}{c}
x_{1} \\
x_{2}
\end{array}
\right] $ is the direct sum of the states $x_{1}$ and $x_{2}$ of the first
and second system respectively, and the cascaded system is then
\begin{eqnarray*}
\left[
\begin{array}{cc}
\left[
\begin{array}{cc}
0 & 0 \\
0 & K_{2}
\end{array}
\right] & \left[
\begin{array}{c}
0 \\
M_{2}
\end{array}
\right] \\
\left[ 0,L_{2}\right] & N_{2}
\end{array}
\right] \ast \left[
\begin{array}{cc}
\left[
\begin{array}{cc}
K_{1} & 0 \\
0 & 0
\end{array}
\right] & \left[
\begin{array}{c}
M_{1} \\
0
\end{array}
\right] \\
\left[ L_{1},0\right] & N_{1}
\end{array}
\right] \\
=\left[
\begin{array}{cc}
\left[
\begin{array}{cc}
K_{1} & 0 \\
M_{2}L_{1} & K_{2}
\end{array}
\right] & \left[
\begin{array}{c}
M_{1} \\
M_{2}L_{1}
\end{array}
\right] \\
\left[ N_{2}L_{1},L_{2}\right] & N_{2}N_{1}
\end{array}
\right] .
\end{eqnarray*}
which gives the correct matrix of coefficients for the systems
\begin{equation*}
\mathbf{V}_{1}\equiv \left\{
\begin{array}{l}
\dot{x}_{1}=K_{1}x_{1}+M_{1}u_{1} \\
y_{1}=L_{1}x_{1}+N_{1}u_{1}
\end{array}
\right. ,\quad \mathbf{V}_{2}\equiv \left\{
\begin{array}{l}
\dot{x}_{2}=K_{2}x_{2}+M_{2}u_{2} \\
y_{2}=L_{2}x_{2}+N_{2}u_{2}
\end{array}
\right. ,
\end{equation*}
under the identification $u_{2}=y_{1}$.

\begin{figure}[htbp]
\centering
\includegraphics[width=0.60\textwidth]{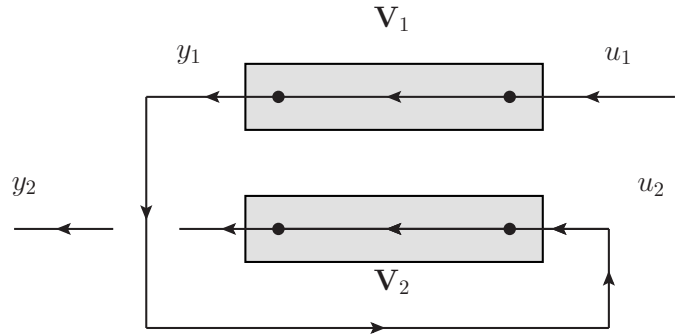}
\caption{Cascaded systems: a special case of the series product where the
inputs $u_1$ and $u_2$ act on separate state variables, that is, distinct
systems.}
\label{fig:lie_trotter_cascade}
\end{figure}

\section{Quantum Stochastic Models}

\subsection{Second Quantization}

We recall the basic ideas of the (Bosonic) second quantization over a
separable Hilbert space $\mathfrak{K}$. The Fock space over $\mathfrak{K}$
is $\Gamma \left( \mathfrak{K}\right) =\oplus _{n=0}^{\infty }\left( \otimes
_{\text{symm.}}^{n}\mathfrak{K}\right) $, and a total set of vectors is
provided by the exponential vectors defined, for test vector $f\in %
\mathfrak{K}$, by
\begin{equation*}
\varepsilon \left( f\right) =1\oplus f\oplus \left( \frac{1}{\sqrt{2!}}%
f\otimes f\right) \oplus \left( \frac{1}{\sqrt{3!}}f\otimes f\otimes
f\right) \oplus \cdots .
\end{equation*}
The creation and annihilation operators with test vector $\phi $ are denoted
as $a^{\dag }\left( \phi \right) $ and $a\left( \phi \right) $ respectively,
and, along with the differential second quantization $d\Gamma \left(
X\right) $ of a self-adjoint operator $X$, they can be defined by
\begin{gather*}
a\left( \phi \right) \varepsilon \left( f\right) =\langle \phi |f\rangle
\varepsilon \left( f\right) ,\quad a^{\dag }\left( \phi \right) \varepsilon
\left( f\right) =\left. \frac{d}{du}\varepsilon \left( f+u\phi \right)
\right| _{u=0}, \\
d\Gamma \left( X\right) \varepsilon \left( f\right) =\left. \frac{1}{i}\frac{%
d}{du}\varepsilon \left( e^{iuX}f\right) \right| _{u=0}.
\end{gather*}
The closures of these operators then satisfy the canonical commutation
relations (CCR) $\left[ a\left( f\right) ,a^{\dag }\left( g\right) \right]
=\langle f|g\rangle $.

\begin{definition}
Let $\mathfrak{K}$\ be a fixed separable Hilbert space. We denote by $U(%
\mathfrak{K}) $ the group of unitary operators on $\mathfrak{K}$ with the
strong operator topology. The Euclidean group EU$\left( \mathfrak{K}\right) $
over $\mathfrak{K}$ is the semi-direct product of $U\left( \mathfrak{K}%
\right) $ with the translation group on $\mathfrak{K}$ and consists of pairs
$\left( T,\phi \right) $ where $T\in U(\mathfrak{K})$ and $\phi \in %
\mathfrak{K}$. The group law is $\left( T_{2},\phi _{2}\right) \circ \left(
T_{1},\phi _{1}\right) =\left( T_{2}T_{1},\phi _{2}+T_{2}\phi _{1}\right) $.
The extended Heisenberg group over $\mathfrak{K}$ is defined to be
\begin{equation*}
\mathfrak{G}\left( \mathfrak{K}\right) =EU\left( \mathfrak{K}\right) \times
\mathbb{R}
\end{equation*}
whose basic elements are triples $\left( T,\phi ,\theta \right) $ with the
group law given by
\begin{equation}
\left( T_{2},\phi _{2},\theta _{2}\right) \vartriangleleft \left( T_{1},\phi
_{1},\theta _{1}\right) =\left( T_{2}T_{1},\phi _{2}+T_{2}\phi _{1},\theta
_{1}+\theta _{2}+\text{Im}\langle \phi _{2}|T_{2}\phi _{1}\rangle \right) .
\label{eq: Heisenberg group law}
\end{equation}
\end{definition}

For $\left( T,\phi \right) \in EU\left( \mathfrak{K}\right) $ we obtain the
Weyl unitary $W\left( T,\phi \right) $ on $\Gamma \left( \mathfrak{K}\right)
$ defined on the domain of exponential vectors by
\begin{equation*}
W\left( T,\phi \right) \,\varepsilon \left( f\right) =\exp \left\{ -\frac{1}{%
2}\left\| \phi \right\| ^{2}-\langle \phi |Tf\rangle \right\} \;\varepsilon
\left( Tf+\phi \right) .
\end{equation*}
The special cases of a pure rotation $\Gamma \left( T\right) =W\left(
T,0\right) $, with $\Gamma \left( e^{iX}\right) =e^{id\Gamma \left( X\right)
}$, and a pure translation $W\left( \phi \right) =W\left( I,\phi \right)
\equiv \exp \left\{ a^{\dag }\left( \phi \right) -a\left( \phi \right)
\right\} $ lead to the second quantization and the Weyl displacement
unitaries respectively. The map $W:$EU$\left( \mathfrak{K}\right) \mapsto
U\left( \Gamma \left( \mathfrak{K}\right) \right) $ however yields only a
projective unitary representation of the Euclidean group since we have
\begin{equation*}
W\left( T_{2},\phi _{2}\right) W\left( T_{1},\phi _{1}\right) =\exp \left\{
-i\text{Im}\langle \phi _{2}|T_{2}\phi _{1}\rangle \right\} \;W\left( \left(
T_{2},\phi _{2}\right) \circ \left( T_{1},\phi _{1}\right) \right) ,
\end{equation*}
which is the Weyl form of the CCR and the presence of the multiplier is
equivalent to the original CCR.

\begin{proposition}
A unitary representation of $\mathfrak{G}\left( \mathfrak{K}\right) $ in
terms of unitaries on the Bose Fock space $\Gamma \left( \mathfrak{K}\right)
$ is then given by the modified Weyl operators $W\left( T,\phi ,\theta
\right) $ with action
\begin{equation*}
W\left( T,\phi ,\theta \right) \,\varepsilon \left( f\right) =e^{-i\theta
}W\left( T,\phi \right) \,\varepsilon \left( f\right) .
\end{equation*}
\end{proposition}

The role of the ``scalar phase'' $\theta $ here is of course to absorb the
Weyl multiplier.

\subsection{Non-abelian Weyl CCR}

We now turn to a question, first posed by Hudson and Parthasarathy in 1983
\cite{HP83}, on how to obtain a non-abelian generalization of the Weyl CCR
version wherein the role of $U\left( 1\right) $ phase is replaced by a
(sub-)group of unitaries $U\left( \mathfrak{h}\right) $ over a fixed
separable Hilbert space $\mathfrak{h}$. In the present paper we show that
the appropriate non-abelian extensions are
\begin{eqnarray*}
T\in U\left( \mathfrak{K}\right) &\rightleftharpoons &S\in U\left( %
\mathfrak{h}\otimes \mathfrak{K}\right) , \\
f\in \mathfrak{K} &\rightleftharpoons &L\in \mathfrak{B}\left( \mathfrak{h},%
\mathfrak{h}\otimes \mathfrak{K}\right) , \\
\theta \in \mathbb{R} &\rightleftharpoons &H\in \mathfrak{B}_{\text{s.a.}%
}\left( \mathfrak{h}\right) ,
\end{eqnarray*}
where $\mathfrak{B}_{\text{s.a.}}\left( \mathfrak{h}\right) $ is the set of
bounded self-adjoint operators on $\mathfrak{h}$. The corresponding law
replacing (\ref{eq: Heisenberg group law}) is the series product:

\begin{definition}
Let $\mathfrak{h}$ and $\mathfrak{K}$\ be a fixed separable Hilbert spaces.
The extended Heisenberg group $\mathfrak{G}\left( \mathfrak{h},\mathfrak{K}%
\right) $ is defined to be the set of triples $\left( S,L,H\right) \in
U\left( \mathfrak{h}\otimes \mathfrak{K}\right) \times \mathfrak{B}\left( %
\mathfrak{h},\mathfrak{h}\otimes \mathfrak{K}\right) \times \mathfrak{B}_{%
\text{s.a.}}\left( \mathfrak{h}\right) $, with group law given by the
(special) series product
\begin{equation}
\left( S_{2},L_{2},H_{2}\right) \vartriangleleft \left(
S_{1},L_{1},H_{1}\right) =\left( S_{2}S_{1},L_{2}+S_{2}L_{1},H_{1}+H_{2}+%
\mathrm{Im}L_{2}^{\dag }S_{2}L_{1}\right) .
\label{eq: series product unitary}
\end{equation}
\end{definition}

Unlike the general situation in quantum groups, the product $%
\vartriangleleft $\ does in fact lead to a group law! It originated in the
work of one of the authors in relation to a systems theoretic approach to\
``cascaded'' quantum stochastic models \cite{QFN1},\cite{GJ Series}.

The original answer provided by Hudson and Parthasarathy involved the
quantum It\={o} calculus with initial space $\mathfrak{h}$ and multiplicity
space $\mathfrak{K}$, see below, in which a triple $\left( S,L,H\right) $
encoded the information on the coefficients of a quantum stochastic
evolution. Apart from a restriction to quantum It\={o} diffusions $\left(
S=I\right) $, they also considered only the operator product of the unitary
quantum evolutions which forced the introduction of time dependence -
effectively the coefficients $\left( S_{1},L_{1},H_{1}\right) $ will be
evolved by the unitary process generated by the second set $\left(
S_{2},L_{2},H_{2}\right) $. The $S\neq I$ case is readily handled with the
aid of quantum stochastic calculus employing the gauge process.

We shall show that\ the natural Lie-Trotter product formula for a pair of
quantum stochastic evolutions leads naturally to the series product (\ref
{eq: series product unitary}), which from the above is the generalization of
the Weyl canonical commutations relations to the non-abelian setting.

\subsection{Quantum Stochastic Evolutions}

We recall the quantum stochastic calculus of Hudson and Parthasarathy \cite
{HP}. The Hilbert space for the system and noise is $\mathfrak{H}=%
\mathfrak{h}\otimes \Gamma \left( L_{\mathfrak{K}}^{2}[0,\infty )\right) $
where $\mathfrak{h}$ is a fixed separable Hilbert space called the initial
space (modelling a quantum mechanical system) and we have the Fock space
over the space of square-integrable $\mathfrak{K}$-valued functions on $%
[0,\infty )$. Note that $L_{\mathfrak{K}}^{2}[0,\infty )\cong \mathfrak{K}%
\otimes L^{2}[0,\infty )$. For transparency of presentation, we restrict to
the case where $\mathfrak{K}$ is $\mathbb{C}^{n}$, however the general case
of a separable Hilbert space presents no difficulties. Let $\left\{
e_{j}\right\} _{j=1}^{n}$ be a basis of $\mathfrak{K}$ (the multiplicity
space) and define the operators
\begin{eqnarray*}
\Lambda ^{00}\left( t\right) &\triangleq &t, \\
\Lambda ^{i0}\left( t\right) &=&A_{i}^{\dag }\left( t\right) \triangleq
a^{\dag }\left( |e_{i}\rangle \otimes 1_{\left[ 0,t\right] }\right) , \\
\Lambda ^{0j}\left( t\right) &=&A_{j}\left( t\right) \triangleq a\left(
|e_{j}\rangle \otimes 1_{\left[ 0,t\right] }\right) , \\
\Lambda ^{ij}\left( t\right) &\triangleq &d\Gamma \left( |e_{i}\rangle
\langle e_{j}|\otimes \chi _{\left[ 0,t\right] }\right) ,
\end{eqnarray*}
where $1_{\left[ 0,t\right] }$ is the characteristic function of the
interval $\left[ 0,t\right] $ and $\chi _{\left[ 0,t\right] }$ is the
operator on $L^{2}[0,\infty )$ corresponding to multiplication by $1_{\left[
0,t\right] }$. Hudson and Parthasarathy developed a quantum It\={o} calculus
where integrals of adapted processes with respect to the fundamental
processes $\Lambda ^{\alpha \beta }$. The It\={o} table is then
\begin{equation*}
d\Lambda ^{\alpha \beta }\left( t\right) \,d\Lambda ^{\mu \nu }\left(
t\right) =\hat{\delta}_{\beta \mu }\,d\Lambda ^{\alpha \nu }\left( t\right)
\end{equation*}
where $\hat{\delta}_{\alpha \beta }$ is the Evans-Hudson delta defined to be
unity if $\alpha =\beta \in \left\{ 1,\cdots ,n\right\} $ and zero
otherwise. This may be written as
\begin{equation*}
\begin{tabular}{l|llll}
$\times $ & $dA_{k}$ & $d\Lambda _{kl}$ & $dA_{l}^{\dag }$ & $dt$ \\ \hline
$dA_{i}$ & 0 & $\delta _{ik}dA_{l}$ & $\delta _{il}dt$ & 0 \\
$d\Lambda _{ij}$ & 0 & $\delta _{jk}d\Lambda _{il}$ & $\delta _{jl}dA_{i}$ &
0 \\
$dA_{j}^{\dag }$ & 0 & 0 & 0 & 0 \\
$dt$ & 0 & 0 & 0 & 0
\end{tabular}
\end{equation*}
In particular, we have the following theorem \cite{HP}.

\begin{theorem}
There exists a unique solution $V\left( \cdot ,\cdot \right) $ to the
quantum stochastic integro-differential equation
\begin{equation}
V\left( t,s\right) =I+\int_{s}^{t}dG\left( \tau \right) \,V\left( \tau
,s\right)  \label{eq diff eq}
\end{equation}
$\left( t\geq s\geq 0\right) $ where
\begin{equation*}
dG\left( t\right) =G_{\alpha \beta }\otimes d\Lambda ^{\alpha \beta }\left(
t\right)
\end{equation*}
with $G_{\alpha \beta }\in \mathfrak{B}\left( \mathfrak{h}\right) $. (We
adopt the convention that we sum repeated Greek indices over the range $%
0,1,\cdots ,n$.)
\end{theorem}

We refer to $\mathbf{G=}\left[ G_{\alpha \beta }\right] \in \mathfrak{B}%
\left( \mathfrak{h}\otimes \left( \mathbb{C}\oplus \mathfrak{K}\right)
\right) $, as the \textit{coefficient matrix}, and $V$ as the left process
generated by $\mathbf{G}$. With respect to the decomposition $\mathfrak{h}%
\otimes \left( \mathbb{C}\oplus \mathfrak{K}\right) =\mathfrak{h}\oplus
\left( \mathfrak{h}\otimes \mathfrak{K}\right) $ we may write
\begin{equation*}
\mathbf{G}=\left[
\begin{array}{cc}
K & M \\
L & N-I
\end{array}
\right]
\end{equation*}
where $K\in \mathfrak{B}\left( \mathfrak{h}\right) ,L\in \mathfrak{B}\left( %
\mathfrak{h},\mathfrak{h}\otimes \mathfrak{K}\right) ,M\in \mathfrak{B}%
\left( \mathfrak{h}\otimes \mathfrak{K},\mathfrak{h}\right) $ and $N\in %
\mathfrak{B}\left( \mathfrak{h}\otimes \mathfrak{K}\right) $. In the
situation where $\mathfrak{K}$ is $\mathbb{C}^{n}$ we have $G_{00}=K$, $L$
is the column vector $\left[ G_{i0}\right] $, $M$ is the row vector $\left[
G_{0j}\right] $ and $N_{ij}=G_{ij}$.

Adopting the convention that repeated Latin indices are summed over the
range $1,\cdots ,n$, we may write in more familiar notation \cite{HP}
\begin{equation*}
dG\left( t\right) =K\otimes dt+M_{i}\otimes dA_{i}\left( t\right)
+L_{j}\otimes dA_{j}^{\dag }\left( t\right) +(N_{ij}-\delta _{ij})\otimes
d\Lambda _{ij}\left( t\right) .
\end{equation*}
For emphasis, we shall often write $V_{\mathbf{G}}\left( \cdot ,\cdot
\right) $ when we wish to emphasize the dependence on the coefficients $%
\mathbf{G}$. We remark that the process satisfies the following properties:

\begin{enumerate}
\item  Flow Law: $V_{\mathbf{G}}\left( t,r\right) V_{\mathbf{G}}\left(
r,s\right) =V_{\mathbf{G}}\left( t,s\right) $ whenever $t\geq r\geq s$.

\item  Stationarity: $\Gamma \left( \theta _{\tau }\right) V_{\mathbf{G}%
}\left( t,s\right) \Gamma \left( \theta _{\tau }\right) =V_{\mathbf{G}%
}\left( t+\tau ,s+\tau \right) $ where $\theta _{\tau }$ is the shift map on
$L_{\mathfrak{K}}^{2}[0,\infty )$.

\item  Localization: with respect to the decomposition $\mathfrak{h}\otimes
\Gamma \left( L_{\mathfrak{K}}^{2}[0,\infty )\right) \cong \mathfrak{h}%
\otimes \Gamma (L_{\mathfrak{K}}^{2}[0,s)\otimes \Gamma \left( L_{%
\mathfrak{K}}^{2}[s,t)\right) \otimes \Gamma \left( L_{\mathfrak{K}%
}^{2}[t,\infty )\right) $, $V_{\mathbf{G}}\left( t,s\right) $ acts trivially
on the factors $\Gamma (L_{\mathfrak{K}}^{2}[0,s)$ and $\Gamma \left( L_{%
\mathfrak{K}}^{2}[t,\infty )\right) $.
\end{enumerate}

It is convenient to introduce the projection matrix (the Hudson-Evans delta)
\begin{equation*}
\mathbf{\hat{\delta}=}\left[
\begin{array}{cc}
0 & 0 \\
0 & I
\end{array}
\right] \equiv \left[ \hat{\delta}_{\alpha \beta }\right] .
\end{equation*}
The key result from \cite{HP} is the following concerning unitary evolutions.

\begin{theorem}
\label{thm: G unitary} Necessary and sufficient conditions on $\mathbf{G}$
to generate a unitary family are that it satisfies the identities
\begin{equation*}
\mathbf{G}+\mathbf{G}^{\dag }+\mathbf{G}^{\dag }\mathbf{\hat{\delta}G}=0%
\text{ \ \ (isometry), }\mathbf{G}+\mathbf{G}^{\dag }+\mathbf{G\hat{\delta}G}%
^{\dag }=0\text{ \ \ (co-isometry),}
\end{equation*}
and this is equivalent to $\mathbf{G}$ taking the form
\begin{equation}
\mathbf{G}_{(S,L,H)}=\left[
\begin{array}{cc}
-\frac{1}{2}L^{\dag }L-iH & -L^{\dag }S \\
L & S-I
\end{array}
\right]  \label{eq:G_unitary}
\end{equation}
with $S$ is a unitary and $H$ is self-adjoint. We then refer to the triple $%
\left( S,L,H\right) $ as Hudson-Parthasarathy coefficients.
\end{theorem}

We shall refer to a coefficient matrix as being a \textit{unitary It\={o}
generator matrix} if it leads to a unitary process. We may likewise consider
right processes, defined as the solution to $U\left( t,s\right)
=I+\int_{s}^{t}U\left( \tau ,s\right) \,dG\left( \tau \right) $, and denote
these as $U_{\mathbf{G}}$. We find that $U_{\mathbf{G}^{\dag }}=V_{\mathbf{G}%
}^{\dag }$. It turns out that it is technically easier to establish
existence of right processes, especially when the $G_{\alpha \beta }$ are
unbounded.

\subsection{The General Series Product}

\begin{definition}
The (general) series product of two coefficient matrices is defined to be
\begin{equation}
\mathbf{G}_{2}\vartriangleleft \mathbf{G}_{1}\triangleq \mathbf{G}_{1}+%
\mathbf{G}_{2}+\mathbf{G}_{2}\mathbf{\hat{\delta}G}_{1}.
\label{eq:series_prod}
\end{equation}
With respect to the standard decomposition above, this corresponds to
\begin{equation}
\left[
\begin{array}{cc}
K_{2} & M_{2} \\
L_{2} & N_{2}-I
\end{array}
\right] \vartriangleleft \left[
\begin{array}{cc}
K_{1} & M_{1} \\
L_{1} & N_{1}-I
\end{array}
\right] =\left[
\begin{array}{cc}
K_{1}+K_{2}+M_{2}L_{1} & M_{1}+M_{2}N_{1} \\
L_{2}+N_{2}L_{1} & N_{2}N_{1}-I
\end{array}
\right] .  \label{eq: series product general}
\end{equation}
\end{definition}

The series product is not commutative, however it is readily seen to be
associative. Let define the \textit{model matrix} $\mathbf{V}$ associated to
a coefficient matrix $\mathbf{G}$ to be
\begin{equation*}
\mathbf{V}\triangleq \mathbf{\hat{\delta}+G}=\left[
\begin{array}{cc}
K & M \\
L & N
\end{array}
\right] .
\end{equation*}

\begin{remark}
The series product $\mathbf{G}_{2}\vartriangleleft \mathbf{G}_{1}$ for two
coefficient matrices implies the corresponding law $\mathbf{V}_{2}\ast
\mathbf{V}_{1}$ for the associated model matrices given by
\begin{equation*}
\mathbf{V}_{2}\ast \mathbf{V}_{1}=\left[
\begin{array}{cc}
K_{1}+M_{2}L_{1}+K_{2} & M_{1}+M_{2}N_{1} \\
L_{2}+N_{2}L_{1} & N_{2}N_{1}
\end{array}
\right] .
\end{equation*}
Note that this is the natural generalization to the rule (\ref
{eq:series_model}) already seen for classical linear state based models in
series!
\end{remark}

\begin{remark}
For It\={o} generating matrices for unitary process we have
\begin{equation*}
\mathbf{G}_{(S_{2},L_{2},H_{2})}\vartriangleleft \mathbf{G}%
_{(S_{1},L_{1},H_{1})}=\mathbf{G}_{(S_{2},L_{2},H_{2})\vartriangleleft
\left( S_{1},L_{1},H_{1}\right) }
\end{equation*}
which again leads to a unitary process. Therefore the general series product
defined in (\ref{eq: series product general}) implies the special series
product (\ref{eq: series product unitary}).
\end{remark}

\begin{lemma}
The increment $dG$ associated with $V_{\mathbf{G}_{2}\vartriangleleft
\mathbf{G}_{1}}\ $is related to the increments $dG_{i}$ associated with $V_{%
\mathbf{G}_{i}}$ through the identity
\begin{equation}
dG=dG_{1}+dG_{2}+dG_{2}dG_{1}  \label{eq:series_increments}
\end{equation}
and this is equivalent to the algebraic relation (\ref{eq:series_prod}) or (%
\ref{eq: series product general}).
\end{lemma}

This follows from a straightforward application of the quantum It\={o}
calculus.

\bigskip

\subsection{The Group of Coefficient Matrices}

\begin{definition}
Denote by GL$_{\vartriangleleft }\left( \mathfrak{h},\mathfrak{K}\right) $
the subset of $\mathfrak{B}\left( \mathfrak{h}\otimes (\mathbb{C}\oplus %
\mathfrak{K})\right) $ consisting of operators of the form
\begin{equation*}
\mathbf{G}=\left[
\begin{array}{cc}
K & M \\
L & N-I
\end{array}
\right]
\end{equation*}
with respect to the decomposition $\mathfrak{h}\oplus (\mathfrak{h}\otimes %
\mathfrak{K})$ of $\mathfrak{h}\otimes (\mathbb{C}\oplus \mathfrak{K})$, and
where $N$ is required to be invertible. GL$_{\vartriangleleft }\left( %
\mathfrak{h},\mathfrak{K}\right) $ becomes a group under the general series
product $\vartriangleleft $ given in (\ref{eq: series product general}).
\end{definition}

We note that the zero operator is the group identity, and that the series
product inverse of $\left[
\begin{array}{cc}
K & M \\
L & N-I
\end{array}
\right] $ is $\left[
\begin{array}{cc}
-K+MN^{-1}L & -MN^{-1} \\
-N^{-1}L & N^{-1}-I
\end{array}
\right] $. The extended Heisenberg group $\mathfrak{G}\left( \mathfrak{h},%
\mathfrak{K}\right) $ is then a subgroup of GL$_{\vartriangleleft }\left( %
\mathfrak{h},\mathfrak{K}\right) $ inheriting the series product as law.

The set $\mathfrak{G}\left( \mathfrak{h},\mathfrak{K}\right) $ was
introduced in \cite{QFN1}\ as the collection of all It\={o} generator
matrices (\ref{eq:G_unitary}) and was shown to be a group under the series
product (\ref{eq: series product unitary}), though not identified as a
Heisenberg group.

\begin{remark}
\label{remark: G inverse} The isometry and co-isometry conditions in theorem
(\ref{thm: G unitary}) imply that a two-sided inverse of $\mathbf{G}\sim
\left( S,L,H\right) \in \mathfrak{G}\left( \mathfrak{h},\mathfrak{K}\right) $
for the series product is given by $\mathbf{G}^{\dag }\sim \left( S^{\dag
},-S^{\dag }L,-H\right) $. The inverse being of course unique.
\end{remark}

\begin{lemma}
The mapping :GL$_{\vartriangleleft }\left( \mathfrak{h},\mathfrak{K}\right)
\mapsto \mathfrak{B}\left( \mathfrak{h}\otimes (\mathbb{C}\oplus \mathfrak{K}%
\oplus \mathbb{C})\right) $ given by
\begin{equation*}
\mathbf{G}=\left[
\begin{array}{cc}
K & M \\
L & N-I
\end{array}
\right] \mapsto \mathbb{V}_{\mathbf{G}}=\left[
\begin{array}{ccc}
I & M & K \\
0 & N & L \\
0 & 0 & I
\end{array}
\right] .
\end{equation*}
is an injective group homomorphism.
\end{lemma}

One readily checks that $\mathbb{V}_{\mathbf{G}_{2}}\mathbb{V}_{\mathbf{G}%
_{1}}=\mathbb{V}_{\mathbf{G}_{2}\vartriangleleft \mathbf{G}_{1}}$, and $%
\mathbb{V}_{\mathbf{G}}^{-1}=\mathbb{V}_{\mathbf{G}^{\dag }}$

This representation is the basis for Belavkin's formalism of quantum
stochastic calculus \cite{Bel98},\cite{TSDoklady}. The Lie algebra of GL$%
_{\vartriangleleft }\left( \mathfrak{h},\mathfrak{K}\right) $ (in the
Belavkin representation) consists of matrices
\begin{equation*}
\mathbb{H}=\left[
\begin{array}{ccc}
0 & \mu & \kappa \\
0 & \nu & \lambda \\
0 & 0 & 0
\end{array}
\right]
\end{equation*}
where now the entries $\kappa ,\lambda ,\mu ,\nu $ are operators and the
exponential map is then $\exp \left( \mathbb{H}\right) =\mathbb{V}_{\mathbf{G%
}}$ with the entries $K,L,M,N$ given by
\begin{equation}
\begin{array}{ll}
K=\kappa +\mu e_{2}\left( \nu \right) \lambda & M=\mu e_{1}\left( \nu
\right) , \\
L=e_{1}\left( \nu \right) \lambda , & N=e^{\nu },
\end{array}
\label{decapitated}
\end{equation}
where we encounter the `decapitated exponential' functions being the entire
analytic functions $e_{1}\left( z\right) =\dfrac{e^{z}-1}{z}$, $e_{2}\left(
z\right) =\dfrac{e^{z}-1-z}{z^{2}}$.

With an abuse of notation we shall take the Lie algebra of GL$%
_{\vartriangleleft }\left( \mathfrak{h},\mathfrak{K}\right) $ to be the
vector space $\mathfrak{gl}_{\vartriangleleft }\left( \mathfrak{h},%
\mathfrak{K}\right) $ of operators $\mathbf{H}=\left[
\begin{array}{cc}
\kappa & \mu \\
\lambda & \nu
\end{array}
\right] $ with entries matched with the representation element $\mathbb{H}$
above and Lie bracket
\begin{equation*}
\left[ \mathbf{H}_{2},\mathbf{H}_{1}\right] =\left[
\begin{array}{cc}
\kappa _{2}\lambda _{1}-\kappa _{1}\lambda _{2} & \mu _{2}\nu _{1}-\mu
_{1}\xi _{2} \\
\nu _{2}\lambda _{1}-\nu _{1}\lambda _{2} & \left[ \nu _{2},\nu _{1}\right]
\end{array}
\right] .
\end{equation*}
With this convention, the exponential map $\widehat{\exp }$ from $%
\mathfrak{gl}_{\vartriangleleft }\left( \mathfrak{h},\mathfrak{K}\right) $
to GL$_{\vartriangleleft }\left( \mathfrak{h},\mathfrak{K}\right) $ takes $%
\mathbf{H}=\left[
\begin{array}{cc}
\kappa & \mu \\
\lambda & \nu
\end{array}
\right] $ to $\mathbf{G}=\left[
\begin{array}{cc}
K & M \\
L & N-I
\end{array}
\right] $ with entries given by (\ref{decapitated}), and this corresponds to
\begin{equation*}
\widehat{\exp }\left( \mathbf{H}\right) \triangleq \sum_{n=1}^{\infty }\frac{%
1}{n!}\mathbf{H}\left( \mathbf{\hat{\delta}H}\right) ^{n-1}.
\end{equation*}

The Lie algebra for the subgroup $\mathfrak{G}\left( \mathfrak{h},%
\mathfrak{K}\right) $ will have elements $\kappa =-i\eta $ and $\nu
=-i\sigma $ with $\eta \in \mathfrak{B}_{\text{s.a.}}\left( \mathfrak{h}%
\right) $ and $\sigma \in \mathfrak{B}_{\text{s.a.}}\left( \mathfrak{h}%
\otimes \mathfrak{K}\right) $, while $\lambda \in \mathfrak{B}\left( %
\mathfrak{h},\mathfrak{h}\otimes \mathfrak{K}\right) $ is arbitrary but with
$\mu =-\lambda ^{\dag }$. The exponential map then leads to the element with
Hudson-Parthasarathy parameters
\begin{equation*}
\left( S,L,H\right) =\left( e^{-i\sigma },e_{1}\left( -i\sigma \right)
\lambda ,\eta +\lambda ^{\dag }\text{Im}\{e_{2}\left( -i\sigma \right)
\}\lambda \right) .
\end{equation*}

\section{Lie-Trotter Formulas}

We set $\Delta ^{2}=\left\{ \left( t,s\right) :t\geq s\geq 0\right\} \subset
\mathbb{R}^{2}$, with each element $\left( t,s\right) \in \Delta ^{2}$
determining an associated interval $[s,t)$ in $[0,\infty )$. Let $%
\mathfrak{A}$ be a *-algebra with a fixed topology, which for concreteness
we may take as acting on some common domain of a Hilbert space.

\begin{definition}
Given an $\mathfrak{A}$-valued function $V\left( \cdot ,\cdot \right) $ on $%
\Delta ^{2}$ we set
\begin{equation}
\left[ V\right] _{\mathcal{P}}\left( t,s\right) \triangleq V\left(
t,t_{n}\right) V\left( t_{n},t_{n-1}\right) \cdots V\left(
t_{2},t_{1}\right) V\left( t_{1},s\right)
\end{equation}
where $\mathcal{P}=\left\{ t>t_{n}>t_{n-1}>\cdots >t_{1}>s\right\} $ is a
partition of the interval $\left[ s,t\right] $. The grid size is $\left|
\mathcal{P}\right| =\max_{k}\left( t_{k+1}-t_{k}\right) $ and we say that
the limit
\begin{equation*}
\lim_{\left| \mathcal{P}\right| \rightarrow 0}\left[ V\right] _{\mathcal{P}%
}\left( t,s\right)
\end{equation*}
exists if $\left[ V\right] _{\mathcal{P}}\left( t,s\right) $ converges in
the topology to a fixed element of $\mathfrak{A}$ independently of the
sequence of partitions used. If the limit is well defined for all $t>s\geq 0$
then we shall write the corresponding two-parameter function as $%
\lim_{\left| \mathcal{P}\right| \rightarrow 0}\left[ V\right] _{\mathcal{P}%
}\left( \cdot ,\cdot \right) $.
\end{definition}

\subsection{Examples}

\subsubsection{Trivial}

If we start with a quantum stochastic exponential $V=V_{\mathbf{G}}$, the
flow property implies that we trivially have $\left[ V_{\mathbf{G}}\right] _{%
\mathcal{P}}\left( t,s\right) =V_{\mathbf{G}}\left( t,s\right) $ for any
partition $\mathcal{P}$.

\subsubsection{Quantum stochastic exponentials}

In the setting of quantum stochastic calculus, we let $G\left( t\right)
=G_{\alpha \beta }\otimes \Lambda ^{\alpha \beta }\left( t\right) $, with $%
G_{\alpha \beta }$ bounded, and set $\left( I+\Delta G\right) \left(
t,s\right) =I+G\left( t\right) -G\left( s\right) $, then
\begin{equation*}
V_{\mathbf{G}}=\lim_{\left| \mathcal{P}\right| \rightarrow 0}\left[ 1+\Delta
G\right] _{\mathcal{P}}.
\end{equation*}

\subsubsection{Holevo's time-ordered exponentials}

In the same setting, we let $H\left( t\right) =H_{\alpha \beta }\otimes
\Lambda ^{\alpha \beta }\left( t\right) $ and set $e^{\Delta H}\left(
t,s\right) =e^{H\left( t\right) -H\left( s\right) }$ then the limit is the
Holevo time-ordered exponential \cite{Holevo92}
\begin{equation*}
\,Y_{\mathbf{H}}=\lim_{\left| \mathcal{P}\right| \rightarrow 0}\left[
e^{\Delta H}\right] _{\mathcal{P}},
\end{equation*}
often written as $\,Y_{\mathbf{H}}\left( t,s\right) =\overleftarrow{\exp }%
\int_{s}^{t}dH\left( \tau \right) $. Holevo established strong convergence
for such limits, including an extension to the situation where $H\left(
t\right) =\int_{0}^{t}H_{\alpha \beta }\left( \tau \right) \otimes d\Lambda
^{\alpha \beta }\left( \tau \right) $ with $H_{\alpha \beta }\left( \cdot
\right) \,$strongly continuous $\mathfrak{B}\left( \mathfrak{h}\right) $%
-valued functions with the $H_{i0}\left( \cdot \right) $ and $H_{0j}\left(
\cdot \right) $ square integrable, and the $H_{ij}\left( \cdot \right) $
integrable.

We should think of the $\mathbf{H=}\left[ H_{\alpha \beta }\right] $ of the
Holevo time-ordered exponential as an element of the Lie algebra $%
\mathfrak{gl}_{\vartriangleleft }\left( \mathfrak{h},\mathfrak{K}\right) $.
In particular, we have the following result.

\begin{lemma}
The Holevo time-ordered exponential $\,Y_{\mathbf{H}}$ is equivalent to the
quantum stochastic exponential $V_{\mathbf{G}}$\ where $\mathbf{G}=\widehat{%
\exp }\left( \mathbf{H}\right) $.
\end{lemma}

\begin{proof}
We observe that the integro-differential equation (\ref{eq diff eq}) can be
given the infinitesimal form
\begin{equation*}
V_{\mathbf{G}}\left( t+dt,t\right) =I+dG\left( t\right)
\end{equation*}
while for the time-ordered exponential we have
\begin{equation*}
\,Y_{\mathbf{H}}\left( t+dt,t\right) =e^{dH\left( t\right) }.
\end{equation*}
For the two to be equal, we need the coefficients of
\begin{equation*}
dG\left( t\right) =dH\left( t\right) +\frac{1}{2!}dH\left( t\right) dH\left(
t\right) +\cdots
\end{equation*}
to coincide, but from the It\={o} table this implies $\mathbf{G}=\widehat{%
\exp }\left( \mathbf{H}\right) $.
\end{proof}

\subsection{The Quantum Stochastic Lie-Trotter Formula}

\begin{definition}
Given $\mathfrak{A}$-valued functions $V_{1}\left( \cdot ,\cdot \right) $
and $V_{2}\left( \cdot ,\cdot \right) $ on $\Delta ^{2}$, we define their
product $V_{2}\cdot V_{1}$ \textit{interval-wise}, that is
\begin{equation}
\left( V_{2}\cdot V_{1}\right) \left( t,s\right) \triangleq V_{2}\left(
t,s\right) V_{1}\left( t,s\right) .  \label{eq: intervalwise product}
\end{equation}
\end{definition}

Note that the product $V_{2}\cdot V_{1}$ will not generally satisfy the flow
property even when $V_{1}$ and $V_{2}$ do, with the result the limit $%
\lim_{\left| \mathcal{P}\right| \rightarrow 0}\left[ V_{2}\cdot V_{1}\right]
_{\mathcal{P}}\left( t,s\right) $ may now not be trivial.

As an example, take the algebra of $n\times n$ matrices $\mathfrak{A}%
=M_{n}\left( \mathbb{C}\right) $ and define $U_{A}\left( t,s\right)
=e^{\left( t-s\right) A}$, then the Lie product formula $\lim_{n\rightarrow
\infty }\left( e^{tA/n}e^{tB/n}\right) ^{n}=e^{t\left( A+B\right) }$ can be
recast in the form
\begin{equation*}
\lim_{\left| \mathcal{P}\right| \rightarrow 0}\left[ U_{A}\cdot U_{B}\right]
_{\mathcal{P}}=U_{A+B}.
\end{equation*}
The extension to the algebra of operators over a Hilbert space with strong
operator topology was subsequently given by Trotter. For instance, if $%
A=-iH_{1}$ and $B=-iH_{2}$ where $H_{1}$ and $H_{2}$ are self-adjoint with $%
H_{1}+H_{2}$ essentially self-adjoint on the overlap of their domains then
the strong limit exists (Theorem VIII.31 \cite{RS I}). The case of strongly
continuous contractive semigroups on Banach spaces is given as Theorem X.5.1
in \cite{RS II}.

We are now able to formulate our main result.

\begin{theorem}
\label{thm: Lie-Trotter Series Product} Let $\mathbf{G}_{1}$ and $\mathbf{G}%
_{2}$ be a pair of bounded coefficient matrices on the same
Hudson-Parthasarathy space, then in the strong operator topology
\begin{equation}
\lim_{\left| \mathcal{P}\right| \rightarrow 0}\left[ V_{\mathbf{G}_{2}}\cdot
V_{\mathbf{G}_{1}}\right] _{\mathcal{P}}=V_{\mathbf{G}_{2}\vartriangleleft
\mathbf{G}_{1}}.  \label{eq: Lie-Trotter series product}
\end{equation}
\end{theorem}

Similarly we find $\lim_{\left| \mathcal{P}\right| \rightarrow 0}\left[ V_{%
\mathbf{G}_{m}}\cdot \ldots \cdot V_{\mathbf{G}_{2}}\cdot V_{\mathbf{G}_{1}}%
\right] _{\mathcal{P}}=V_{\mathbf{G}_{m}\vartriangleleft \cdots
\vartriangleleft \mathbf{G}_{2}\vartriangleleft \mathbf{G}_{1}}$, where the
interval-wise multiple products are defined in the obvious way.

\begin{proof}
To see where this comes from, we note from the infinitesimal form that $%
V=\lim_{\left| \mathcal{P}\right| \rightarrow 0}\left[ V_{\mathbf{G}%
_{2}}\cdot V_{\mathbf{G}_{1}}\right] _{\mathcal{P}}$ should satisfy the
analogous equation
\begin{equation*}
V\left( t+dt,t\right) =\left( I+dG_{2}\left( t\right) \right) \left(
I+dG_{1}\left( t\right) \right) \equiv I+dG\left( t\right)
\end{equation*}
where $dG=dG_{1}+dG_{2}+dG_{2}dG_{1}$, but by (\ref{eq:series_increments})
we recognize this as just the infinitesimal generator of $V_{\mathbf{G}%
_{2}\vartriangleleft \mathbf{G}_{1}}$. In contrast to the traditional
Lie-Trotter formulas, the above limit depends on the order of $V_{\mathbf{G}%
_{2}}\cdot V_{\mathbf{G}_{1}}$ and is therefore asymmetric under interchange
of $V_{\mathbf{G}_{2}}$ and $V_{\mathbf{G}_{1}}$.
\end{proof}

\subsection{Special Cases}

\subsubsection{Lie-Trotter formula}

The special case $\mathbf{G}_{i}=\left[
\begin{array}{cc}
K_{i} & 0 \\
0 & 0
\end{array}
\right] $ recovers the usual Lie-Trotter formulas.

\subsubsection{Separate Channels}

Let $\mathbf{G}_{i}=\left[
\begin{array}{cc}
K_{i} & M_{i} \\
L_{i} & N_{i}-I
\end{array}
\right] $ be coefficient matrices with common initial space $\mathfrak{h}$
but different multiplicity spaces $\mathfrak{K}_{i}$. We combine the
multiplicity space into a single space $\mathfrak{K}=\mathfrak{K}_{1}\oplus %
\mathfrak{K}_{2}$ and ampliate both coefficient matrices as follows:
\begin{equation*}
\mathbf{\tilde{G}}_{1}=\left[
\begin{array}{ccc}
K_{1} & M_{1} & 0 \\
L_{1} & N_{1}-I_{1} & 0 \\
0 & 0 & 0
\end{array}
\right] ,\mathbf{\tilde{G}}_{2}=\left[
\begin{array}{ccc}
K_{2} & 0 & M_{2} \\
0 & 0 & 0 \\
L_{2} & 0 & N_{2}-I_{2}
\end{array}
\right]
\end{equation*}
then
\begin{equation*}
\mathbf{\tilde{G}}_{2}\vartriangleleft \mathbf{\tilde{G}}_{1}=\left[
\begin{array}{ccc}
K_{1}+K_{2} & M_{1} & M_{2} \\
L_{1} & N_{1}-I_{1} & 0 \\
L_{2} & 0 & N_{2}-I_{2}
\end{array}
\right] .
\end{equation*}
The right hand side is taken as the definition of the concatenation\ $%
\mathbf{G}_{1}\boxplus \mathbf{G}_{2}$ of the two separate coefficient
matrices: this is consistent with the definition of concatenation introduced
earlier for model matrices. Theorem (\ref{thm: Lie-Trotter Series Product})
then implies that
\begin{equation*}
\lim_{\left| \mathcal{P}\right| \rightarrow 0}\left[ V_{\mathbf{\tilde{G}}%
_{2}}\cdot V_{\mathbf{\tilde{G}}_{1}}\right] _{\mathcal{P}}=V_{\mathbf{G}%
_{2}\boxplus \mathbf{G}_{1}}.
\end{equation*}
This is equivalent to the result derived by Lindsay and Sinha \cite{LS10}. We should also mention
the recent work of Das, Goswami and Sinha indicates that the Trotter formula should also hold
at the level of flows \cite{DGS10}.

\bigskip

\textbf{Acknowledgement:} The authors are grateful to Professor Kalyan Sinha
for presenting his work on quantum stochastic Lie-Trotter formula \cite{LS10}
during a visit to Aberystwyth within the framework of the UK-India Education
and Research Initiative.


\begin{thebibliography}{99}
\bibitem{HP}  [1] R.L. Hudson and K.R. Parthasarathy, \textit{Quantum Ito's
formula and stochastic evolutions,} Commun. Math. Phys. \textbf{93}, 301-323
(1984)

\bibitem{HP83}  [2] R.L. Hudson and K.R. Parthasarathy, \textit{Generalized Weyl
operators}, in Stochastic Analysis and Applications, Lecture Notes in
Mathematics 1095, Ed. A. Truman, D. Williams, Springer (1983)


\bibitem{LS10}  [3] J.M. Lindsay, K.B. Sinha, \textit{A quantum stochastic
Lie-Trotter product formula}, Indian J. Pure Appl. Math., \textbf{41}%
(1):313-325, February (2010)

\bibitem{QFN1}  [4] J. Gough, M.R. James, \textit{Quantum\ Feedback Networks:
Hamiltonian Formul- ation}, Commun. Math. Phys. \textbf{287}, 1109-1132
(2009)

\bibitem{GJ Series}  [5] J. Gough, M.R. James, \textit{The series product and
its application to feedforward and feedback networks}, IEEE Trans. Autom.
Control \textbf{54}, 2530 (2009)

\bibitem{Holevo92}  [6] A.S. Holevo, \textit{Time-ordered exponentials in
quantum stochastic calculus}, Quantum Probability and Related Topics, Vol.
VII, 175-202 (1992)

\bibitem{Bel98}  [7] V.P. Belavkin: \textit{On quantum Ito algebras and their
decompositions}, Letters in Mathematical Physics \textbf{45}, 131-145 (1998)

\bibitem{TSDoklady}  [8] O.G. Smolyanov, A. Truman, \textit{The Gough-James
model of quantum feedback networks in the Belavkin representation},
(Russian) Dokl. Akad. Nauk \textbf{435}, no. 5, 591-594 (2010)

\bibitem{Partha}  [9] K.R. Parthasarathy, An Introduction to Quantum Stochastic
Calculus, Birkhauser, Berlin, (1992)

\bibitem{EH88}  [10] M.P. Evans, R.L. Hudson, \textit{Multidimensional quantum
diffusions}, Springer LNM 1303, 69-88 (1988)

\bibitem{RS I}  [11] M. Reed and B. Simon, Methods of Modern Mathematical Physics
I: Functional Analysis, Academic Press, New York (1975)

\bibitem{RS II}  [12] M. Reed and B. Simon, Methods of Modern Mathematical
Physics II: Analysis of Operators, Academic Press, New York (1975)

\bibitem{DGS10} [13] B. Das, D. Goswami and K.B. Sinha,
\textit{A Trotter product formula for quantum stochastic flows}
arXiv:1001.0233v1
\end{thebibliography}
\end{document}